\documentclass[11pt]{article}
\usepackage{verbatim,amssymb,amsthm,amsfonts}
\usepackage{xfrac}
\usepackage[nonamelimits]{amsmath}
\usepackage{fullpage}
\usepackage[dvips,colorlinks=true,linkcolor=black, citecolor=black]{hyperref}

\newcommand\F{\mathbb F}

\def\B{\{0,1\}}

\def\calC{\mathcal{C}}

\providecommand\abs[1]{\lvert#1\rvert}

\def\poly{\mathrm{poly}}

\theoremstyle{plain}

\newtheorem{theorem}{Theorem}[section]

\newtheorem{proposition}[theorem]{Proposition}
\newtheorem{claim}[theorem]{Claim}
\newtheorem{conjecture}[theorem]{Conjecture}
\newtheorem{fact}[theorem]{Fact}

\theoremstyle{definition}
\newtheorem{definition}[theorem]{Definition}

\def\hypref#1#2{\hyperref[#2]{#1~\ref*{#2}}}

\DeclareMathOperator\pr{\mathrm{Pr}}

\def\eps{\varepsilon}

\def\bone{\mathbf{1}}

\def\Enc{\mathbf{Enc}}
\def\Dec{\mathbf{Dec}}
\def\Gen{\mathbf{Gen}}
\def\Eval{\mathbf{Eval}}
\def\HDec{\mathbf{HDec}}
\def\ReEnc{\mathbf{ReEnc}}
\def\Boost{\mathbf{Boost}}
\def\KK{\mathbf{K}}
\def\EE{\mathbf{E}}

\def\TT{\mathbf{T}}
\def\HOM{\mathbf{HOM}}
\def\BASIC{\mathbf{BASIC}}

\title{Homomorphic encryption from codes}
\author{%
Andrej Bogdanov\thanks{{\tt andrejb@cse.cuhk.edu.hk}. Department of Computer Science and Engineering and Institute for Theoretical Computer Science and Communications, Chinese University of Hong Kong. Work supported by RGC GRF CUHK410309.} \and %
Chin Ho Lee\thanks{{\tt chlee@cse.cuhk.edu.hk}. Department of Computer Science and Engineering, Chinese University of Hong Kong.}}
\date{}

\begin{document}
\maketitle

\begin{abstract}
We propose a new homomorphic encryption scheme based on the hardness of decoding under independent random noise from certain affine families of codes. Unlike in previous lattice-based homomorphic encryption schemes, where the message is hidden in the noisy part of the ciphertext, our scheme carries the message in the affine part of the transformation and applies noise only to achieve security. Our scheme can tolerate noise of arbitrary magnitude, as long as the noise vector has sufficiently small hamming weight (and its entries are independent).

Our design achieves ``proto-homomorphic'' properties in an elementary manner: message addition and multiplication are emulated by pointwise addition and multiplication of the ciphertext vectors. Moreover, the extremely simple nature of our decryption makes the scheme easily amenable to bootstrapping. However, some complications are caused by the inherent presence of noticeable encryption error. Our main technical contribution is the development of two new techniques for handling this error in the homomorphic evaluation process.

We also provide a definitional framework for homomorphic encryption that may be useful elsewhere.
\end{abstract}

\section{Introduction}

Homomorphic encryption was proposed by Rivest, Adleman, and Dertouzos~\cite{RAD78} over three decades ago as a mechanism for secure delegation of computation to an honest but curious server. While some partial progress was made over time, the first such cryptographic schemes were proposed only a few years ago, starting with the breakthrough work of Gentry~\cite{Gen09a, Gen09b}. 

Since then, several such schemes have been proposed~\cite{DGHV10, BV11, GH11, BGV12}. These schemes vary widely in their underlying security assumptions as well as the simplicity and efficiency of the constructions. However at a fundamental level, they all rely on the same idea of hiding information inside the noise of lattice-based encryptions. 

We propose a new way to achieve homomorphic encryption based on codes rather than lattices. In both code and lattice based cryptosystems, encryptions are obtained by applying an affine transformation to an input and adding some noise. The two differ in the way they encode information. In lattice based cryptography, the information is encoded inside the noise and the security of the system relies on the inability to distinguish different noise patterns. In code-based cryptography, the information is encoded in the input to the affine transformation, while the role of the noise is to prevent its inversion (and more generally deducing various properties of the input).

\section{Our cryptosystem}
\label{sec:basicsystem}

Our main result is a construction of a homomorphic public-key encryption scheme from a code-based public-key encryption scheme with some special properties. The code-based scheme which is the base of our construction is new. We arrived at it by combining the structure of encryptions of the local cryptosystem of Applebaum, Barak, and Wigderson~\cite{ABW09} with a ``key scrambling'' idea of the McEliece cryptosystem~\cite{McE78}. We begin by discussing the proposed scheme and give evidence in favor of its security. The design is motivated by certain algebraic requirements that enable the implementation of homomorphic operations. We defer the discussion of these special properties to Section~\ref{sec:proto}.

\subsection{The base cryptosystem $K$}

The ciphertexts in our cryptosystem are $n$-bit vectors over $\F_q$, where $q$ is a power of a prime. Three additional parameters that enter the description of the cryptosystem are the amount of randomness $r$ used in the encryption, the size $s$ of the secret key, and the noise distribution $\tilde{\eta}$ over $\F_q$. We will discuss the relationships between these parameters shortly. Conjecture~\ref{conj:security} at the end of this section summarizes the conclusion of this discussion. The message set of our encryption scheme is the set $\F_q$.

\vspace{\baselineskip}
\noindent{\bf Public-key encryption scheme $\KK$}
\smallskip
\hrule
\smallskip
\noindent{\bf Key generation}: Choose a uniformly random subset $S \subseteq \{1, \dots, n\}$ of size $s$ and an $n \times r$ matrix $M$ from the following distribution. First, choose a set of uniformly random but distinct values $a_1, \dots, a_n$ from $\F_q$. Set the $i$th row $M_i$ to
\[ M_i = \begin{cases}
[a_i\ a_i^2\ \cdots\ a_i^{s/3}\ 0\ \cdots\ 0], &\text{if $i \in S$}, \\
[a_i\ a_i^2\ \cdots\ a_i^{s/3}\ a_i^{s/3+1}\ \cdots\ a_i^{r}], &\text{if $i \not\in S$}.
\end{cases} \]
The secret key is the pair $(S, M)$ and the public key is the matrix $P = MR$, where $R$ is a random $r \times r$ matrix over $\F_q$ with determinant one. (Such a matrix can be efficiently sampled.)

\smallskip
\noindent{\bf Encryption}: Given a public key $P$, to encrypt a message $m \in F_q$, choose a uniformly random $x \in \F_q^r$ and a noise vector $e \in \F_q^n$ by choosing each of its entries independently at random from $\tilde{\eta}$. Output the vector $Px + m\bone + e$, where $\bone \in \F_q^n$ is the all ones vector.

\smallskip
\noindent{\bf Decryption}: Given a secret key $(S, M)$, to decrypt a ciphertext $c \in \F_q^n$, first find a solution to the following system of $s/3 + 1$ linear equations over variables $y_i \in \F_q, i \in S$
\begin{equation}
\label{eqn:decrypt}
\begin{array}{rl}
\sum_{i \in S} y_iM_i &= 0 \\
\sum_{i \in S} y_i &= 1
\end{array}
\end{equation}
with $y_i = 0$ when $i \not\in S$. Output the value $\sum_{i \in [n]} y_ic_i$.
\smallskip
\hrule
\vspace{\baselineskip}

To understand the functionality of this scheme, let us first assume that no noise is present, that is $\tilde{\eta}$ always outputs zero. The decryption of an encryption of $m$ is given by 
\[ y^T (Px + m\bone) = (y^T M)Rx + m \cdot y^T\bone = \Bigl(\sum\nolimits_{i \in S} y_iM_i\Bigr)Rx + m\sum\nolimits_{i \in S} y_i = m \]
by the constraints (\ref{eqn:decrypt}) imposed on $y_i$. We must argue that these constraints can be simultaneously satisfied. This follows from the fact that the matrix specifying the system of equations (\ref{eqn:decrypt}) is an $s \times (s/3 + 1)$ Vandermonde matrix, which has full rank and is therefore left-invertible.

When noise is present in the encryption, the decryption could produce the wrong answer when at least one of the noisy elements makes it inside the hidden set $S$. By a union bound this happens with probability at most $\eta s$, where $\eta = \pr[\tilde{\eta} \neq 0]$ is the noise rate of the scheme.

\subsection{Relation with the McEliece and ABW cryptosystems}
While we are unable to argue the security of our proposed scheme by formal reduction to a previously studied one, we describe how our scheme combines ideas from the existing cryptosystems of McEliece and Applebaum, Barak, and Wigderson (ABW), with an eye towards inheriting the security features of these schemes. We take some small liberties in our discussion of these encryption schemes in order to emphasize the parallels to our proposed scheme. 

In the McEliece cryptosystem based on the Reed-Solomon code, the public key looks exactly like in our scheme, except that the secret subset $S$ is empty (i.e., $s = 0$). The syntax and semantics of the encryption, however, are somewhat different. The message set is $\F_q^r$ and an encryption of a message $x \in \F_q^r$ has the form $Px + e$, which looks like a noisy codeword of the Reed-Solomon code.\footnote{One security issue is that these ciphertexts are not message indistinguishable.} Decryption is performed by applying an error-correction algorithm to this codeword. What prevents the adversary from applying the error-correction himself is the fact that the (randomized) evaluation points of the Reed-Solomon code are not revealed in the public key, owing to the presence of the ``key scrambling'' matrix $R$.

In our proposed cryptosystem, the vector $x \in \F_q^r$ does not represent the message but is used to randomize the encryption. Since $P$ and $M$ are generator matrices of the same linear code, the encryption of a message $m \in \F_q$ can be viewed as an affine shift of a random codeword of this code by $m$ units in every coordinate. To thwart decoding by inverting this affine transformation, a noise is injected into some of the coordinates. The ability to decrypt now relies not on the existence of efficient error-correction for the Reed-Solomon code, but on the trapdoor $S$. The submatrix $M_S$ of $M$ indexed by the rows of $S$ has a similar structure to the whole matrix $M$, but on a smaller scale. The scale $s$ of this ``self-similarity'' will be chosen small enough so that noise is unlikely to make it into the codeword coordinates indexed by $S$, allowing for very simple decoding via linear algebra. 

Thus at a structural level, our proposed cryptosystem is quite similar to the ABW cryptosystem. Besides the superficial difference that the ABW system operates over the field $\F_2$ while our system will be instantiated over a larger field, the main difference is in the choice of the public key matrix $P$. In the ABW system, the choice of this matrix is constrained by the fact that the encoding needs to be performed in a local manner. In our case, we will need $M$ (and therefore $P$) to have specific algebraic structure that enables homomorphic operations.

\subsection{Parameters and security}

We now turn to arguing the security of our scheme against certain natural attacks. The form of security that we aim to achieve is the standard notion of $(s, \eps)$ (key independent) message indistinguishability, which asks that for every pair of messages $m, m' \in F_q$, the encryptions of $m$ and $m'$ are indistinguishable with advantage $\eps$ by circuits of size $s$ that are given the public key, where the randomness is taken over the choice of keys.\footnote{Security can be proved even if $m$ and $m'$ are allowed to depend on the public key, but to avoid some technical complications in the definitions we present our results with respect to the weaker notion.}

We describe the attacks at a somewhat informal level in order to gain intuition about the setting of parameters $n$, $q$, $r$, $s$, and $\eta$ for which the proposed scheme could be secure. For convenience in further discussion, $n$ will play the role of a security parameter and we propose values for the other parameters in terms of $n$. Ultimately all of these parameters will be polynomially related to $n$; the exact polynomial dependencies, which are chosen with some foresight, are described by a constant $\alpha > 0$, whose significance will become apparent in Section~\ref{sec:keylength}.

\medskip\noindent{\bf Recover the hidden subset $S$ from the public key.} A natural attack for the adversary is to locate or guess the hidden subset $S$. A brute-force search would go over all $\binom{n}{s}$ possible candidates for $S$. To obtain non-negligible security, one should choose $s$ to increase asymptotically with $n$.

Here is a more sophisticated kind of attack that attempts to obtain information about $S$. A statistical way to distinguish the rows of $P$ that are indexed by $S$ from the other ones is based on the dimension of the hidden vectors in the matrix $P$. For the purposes of describing this attack we can pretend that $P = M$, as the attack only relies on the column space of $P$, which is identical for the two matrices. One can attempt to locate the rows in $M_S$ by calculating the rank of various $k \times r$ submatrices $D$ of $M$. If $D$ turns out not to be of full rank, then $D$ must contain a vector in $S$ (for otherwise $D$ would be a Vandermonde matrix and therefore of full rank). By performing such rank calculations one could expect to find information about the subset $S$.

In Appendix~\ref{app:fullrank} we show that for any $t \times r$ submatrix $D$ (depending on $S$) the rank of $D$ is full with probability at least $1 - O(r^2/q)$, unless $D$ contains at least $s/3 + 1 + \max\{t-r, 0\}$ rows from $M_S$. The probability is taken over the random choice of $a_1, \dots, a_n$ in the key generation algorithm. A simple calculation shows that if $D$ were chosen at random (for any choice of $t$), it would be rank deficient with probability at most $\min\{O(r^2/q), 1/\binom{n}{\Omega(s)}\}$.

Specifically, if we set $s = n^{\alpha/4}$ and $q$ on the order of $2^{n^{\alpha}}$, both of these attacks will require exponential time, or only yield inverse exponential success probability .

\medskip\noindent{\bf Exploit the special properties of $M_S$ in the public key.} In our decryption algorithm it was crucial that the rows of the matrix $M_S$ satisfy the constraints of the linear system (\ref{eqn:decrypt}). However this special structure of $M_S$ could be potentially exploited by an adversary. For instance, an adversary may set up a system of equations analogous to (\ref{eqn:decrypt}), but over all indices of the ciphertext instead only of those in $S$. Specifically, the adversary sets up the following system of equations over variables $y_i, i \in [n]$:
\[ \begin{array}{rl}
\sum_{i \in [n]} y_iP_i &= 0 \\
\sum_{i \in [n]} y_i &= 1.
\end{array} \]
Notice that the solution space of this system does not change if $P$ is replaced by $M$, and so in particular it contains all the solutions to the system (\ref{eqn:decrypt}) (with $y_i = 0$ for $i \not\in S$). If the adversary is lucky, the solution space will contain {\em only} the solutions to (\ref{eqn:decrypt}) so by solving the system he would gain the ability to decrypt.

By choosing $r$ to be sufficiently smaller than $n$---we set $r = n^{1 - \alpha/8}$---we can ensure that the system set up by the adversary has abundantly many solutions, most of which will be forced to have very large hamming weight. Such solutions are useless for the decoding, as long as $\eta$ is not trivially small, because the noise in the ciphertext is likely to affect some nonzero coordinates of $y$.

Our homomorphic algorithms rely on one additional property of the matrix $M_S$, namely the existence of solutions to the more constrained linear system (\ref{eqn:homdecrypt}) described in Section~\ref{sec:proto}. We can argue that the analogous attack fails by a similar argument as to the one given here. Generally, our intuition is that we can handle attacks that exploit the similarity between the matrices $M$ and (the nonzero part of) $M_S$ by choosing the rows-to-columns aspect ratio of $M$ to be substantially larger than the rows-to-columns aspect ratio of $M_S$, which is constant.

\medskip\noindent{\bf Recover the randomness $x$ used in the encryption.} If the noise rate $\eta$ in the encryption is too small, the adversary may be able to recover $x$ from, say, an encryption of $0$.
For instance, if the noise rate $\eta$ is smaller than $1/r$, then in an encryption of $0$ of the form $Px + e$ it would happen with constant probability that no noise makes it into the first $r$ bits of the encryption. In that case, the adversary could recover the randomness by inverting the first $r$ bits of the ciphertext. 

We set the noise rate $\eta$ to $1/n^{1 - \alpha/4}$. Since $r = n^{1 - \alpha/8}$, it follows that any projection of the bits of a ciphertext of linear length is likely to contain noise, which would make it exponentially hard to recover the randomness $x$.

Taking all these factors into consideration, we are now ready to conjecture the security of our proposed cryptosystem $\KK$.

\begin{conjecture}
\label{conj:security}
For every $\alpha > 0$ there exists $\gamma > 0$ such that the cryptosystem $\KK$ with parameters $r = n^{1 - \alpha/8}$, $\eta = 1/n^{1 - \alpha/4}$, $s = n^{\alpha/4}$ and $q \geq 2^{n^\alpha}$ is $(2^{n^\gamma}, 2^{-n^\gamma})$-message indistinguishable, for all $n$ that are sufficiently large.
\end{conjecture}

We will use $\KK_q(n)$ to denote an instantiation of the cryptosystem $\KK$ with the parameters from Conjecture~\ref{conj:security} (except for $q$ which we leave as a free parameter).

\subsection{Our main result}

For technical simplicity we state our definitions and results in the non-uniform setting. An extension to the uniform setting, which is more natural for homomorphic encryption, is straightforward. We chose to work in the simpler non-uniform setting in order to avoid distracting technical and notational complications. 

In our definition of homomorphic encryption we wish to distinguish between the standard decryption algorithm, which applies to encryptions of bits, and the homomorphic decryption algorithm, which applies to the output of the homomorphic evaluation circuit. Also, unlike previous homomorphic encryption schemes, ours carries the risk of a setup error, which we account for in the definition.

Owing to this risk of error, it is possible that some of the inputs provided to the homomorphic evaluation circuit are themselves corrupted. To provide for this possibility, we give a somewhat more general definition of homomorphic evaluation: Instead of requiring that the circuit works well on {\em encryptions} of the inputs (which are not even well-defined in the setting of error-prone probabilistic encryption), we ask that they work on inputs that {\em decrypt} to the correct value. This feature of the definition will be very useful in the proofs.

\begin{definition}
A {\em homomorphic encryption scheme} with setup error $\kappa$ for circuit class $\calC = \{C\colon B^m \to B\}$  
(where $B$ is a subset of the message set) consists of five circuits $(\Gen, \Enc, \Dec,\allowbreak \Eval, \HDec)$, where $(\Gen, \Enc, \Dec)$ is a (probabilistic) public-key encryption scheme (for a formal definition see e.g.~\cite{GolBook}), and $\Eval$ and $\HDec$ are (deterministic) circuits that satisfy
\[ \pr[\HDec_{SK}(\Eval_{PK}(C, c_1, \dots, c_m)) = C(m_1, \dots, m_m)] \geq 1 - \kappa \]
for every circuit $C \in \calC$, every message $m \in \B^m$, and every collection of ciphertexts $c_1, \dots, c_m$ such that $\Dec_{SK}(c_i) = m_i$ for every $i$. The probability is taken over the choice of keys $(SK, PK) \sim \Gen$.
\end{definition}

Let $C\colon \B^m \to \B$ be a boolean circuit with binary addition (i.e. XOR) and multiplication (i.e. AND) gates of fan-in two. The {\em depth} of $C$ is the maximum number of gates on a directed path of $C$. We let $\calC_{cs, d}$ denote the class of such circuits with circuit size $cs$ and depth $d$.

Our main result is a construction of a ``layered'' homomorphic encryption scheme $\HOM$ based on $\KK$, which is fully described in Section~\ref{sec:hom}. The following theorem summarizes the functionality and security properties of our scheme. The parameter $k$ controls the setup error and can be instantiated to any desired value.

\begin{theorem}
\label{thm:main}
Let $q \leq 2^n$ be a power of two. Assume that the public-key encryption $\KK_q(n)$ is $(s(n), \eps(n))$-message indistinguishable for every $n$ (where $s(n)$ and $1/\eps(n)$ are nondecreasing functions of $n$). Then $\HOM$ is a $(s(n^{0.1}) - dk\cdot\poly(n), O(dkn^{1.8}\eps(n^{0.1})))$-message indistinguishable homomorphic encryption scheme for $\calC_{cs, d}$  with key length at most $O(dkn)$, encryption length $O(kn)$, encryption error $2^{-\Omega(k)}$, and setup error $d \cdot 2^{-\Omega(k)}$.
\end{theorem}

\subsection{Overview of $\HOM$}

To begin, in Section~\ref{sec:proto} we show that the operations of {\em pointwise} addition and multiplication already enjoy certain ``proto-homomorphic'' properties, which are sufficient to handle one layer of homomorphic multiplications. We formalize these properties using the new notion of {\em encryption spaces}, which may be a convenient conceptual tool for studying the functionality of homomorphic encryptions. The analysis relies on the special structure of the matrix $M$, specifically on the large redundancy of the constraint system (\ref{eqn:homdecrypt}). 

In Section~\ref{sec:reenc} we give a formal definition of reencryption, a notion crucial (in ours as well as other) constructions. We prove that proto-homomorphic operations together with secure reencryption gives secure homomorphic schemes. We apply an idea of Gentry to obtain a reencryption for our public-key scheme $\KK$. Unfortunately, owing to the inherent noise in our encryptions, the reencryption substantially increases the length of ciphertexts, and the resulting homomorphic scheme has a noticeable setup error.

Section~\ref{sec:optimize} contains the main technical contributions of our work which address these deficiencies. We first give a secure length-preserving reencryption based on a recursive application of the length-increasing reencryption from Section~\ref{sec:reenc} which we use to obtain homomorphic noise correction. We then give a generic mechanism for reducing the setup error, which extends von Neumann's method of building reliable circuits from unreliable components~\cite{Neu56} to the homomorphic setting.

Combining these results, we give the construction of $\HOM$ and prove Theorem~\ref{thm:main} in Section~\ref{sec:hom}.

\section{Encryption spaces and proto-homomorphic operations}
\label{sec:proto}

Since homomorphism of encryptions is a functionality rather than a security requirement, we feel that it is useful to decouple the functionality and security properties of the schemes under discussion. For this purpose we introduce the notion of an {\em encryption space} which is concerned with the set-theoretic properties of encryptions and abstracts away their statistical properties.

\begin{definition}
An {\em encryption space} over message set $\Sigma$ and ciphertext set $\Xi$ is a triple $(Keys, Enc, Dec)$, where
\begin{itemize}
\item $Keys$ is a set of admissible key pairs $(PK, SK)$, 
\item $Enc_{PK}(\cdot)$ is a function that maps messages $m \in \Sigma$ into subsets of valid encyptions $Enc_{PK}(m) \subseteq \Xi$, and 
\item $Dec_{SK}(\cdot)$ is a function that maps messages $m \in \Sigma$ into mutually disjoint valid decryptions $Dec_{SK}(m) \subseteq \Xi$.
\end{itemize}
with the property that $Enc_{PK}(m) \subseteq Dec_{SK}(m)$ for every $(PK, SK) \in Keys$ and $m \in \Sigma$.
\end{definition}

We will say that a public-key encryption scheme $(\Gen, \Enc, \Dec)$ {\em implements} the encryption space $(Keys, Enc, Dec)$ with encryption error $\delta$ if (1) The support of the output distribution of $\Gen$ is contained in $Keys$; (2) For every $m$ and $PK$, $\pr[\Enc_{PK}(m) \in Enc_{PK}(m)] \geq 1 - \delta$; and (3) For every $SK$ and $c \in Dec_{SK}(m)$, $\Dec_{SK}(c) = m$.

\paragraph{An encryption space for $\KK$} Notice that for the functionality of the scheme $\KK$, it only matters what happens to the part of the ciphertext that falls inside the hidden subset $S$. Our definition of the encryption space $K = (Keys, Enc, Dec)$ for $\KK$ will capture this intuition. However, we will equip $K$ with an additional property which will be crucial to achieve proto-homomorphic encryption.

We set $Keys$ to be the support of the key generation algorithm $\Gen$ and $Enc_{PK}(m)$ to be the set of all ciphertexts that take value $Mx + m{\bf 1} + f$, where $f_i = 0$ when $i \in S$ and $f_i$ can be arbitrary when $i \not\in S$. We define $Dec_{SK}(m)$ as the collection of all ciphertexts $c$ that satisfy $y^Tc = m$ for some arbitrary but fixed $y$ that solves the following system of linear equations:
\begin{equation}
\label{eqn:homdecrypt}
\begin{array}{rl}
\sum_{i \in S} y_i(M_i \otimes M_i) &= 0 \\
\sum_{i \in S} y_iM_i &= 0 \\
\sum_{i \in S} y_i &= 1
\end{array}
\end{equation}
with $y_i = 0$ when $i \not \in S$. Here $M_i \otimes M_i$ denotes the tensor product of $M_i$ with itself, which we view as an $s^2$-dimensional vector (after removing the zero entries) whose $(j, k)$th entry is $a_i^ja_i^k = a_i^{j+k}$. Notice that the system (\ref{eqn:homdecrypt}) is more constrained than the system (\ref{eqn:decrypt}) as it includes additional equations. These equations will play a crucial role in enabling homomorphic multiplication.

\begin{claim}
$K$ is an encryption space over message set $\F_q$.
\end{claim}
\begin{proof}
To make sense of the definition of $K$ we must first argue that the system (\ref{eqn:homdecrypt}) has at least one solution $y$. Here is where the structure of the Reed-Solomon code comes in handy: Although the system (\ref{eqn:homdecrypt}) has as many as $s^2$ equations, they all repeat the following set of $2s/3 + 1$ equations:
\[ \begin{array}{rl}
\sum_{i \in S} y_i a_i^k &= 0 \quad\text{for $k = 1, 2, \dots, 2s/3$} \\
\sum_{i \in S} y_i &= 1.
\end{array} \]
The matrix of this system is an $s \times (2s/3 + 1)$ Vandermonde matrix and is therefore left-invertible, so the system is guaranteed to have a solution.

The disjointness of the sets $Dec_{SK}(m)$ is immediate. We now show that $Enc_{PK}(m) \subseteq Dec_{SK}(m)$ for every $m \in \F_q$. Let $c$ be of the form $Mx + m{\bf 1} + f$ and let $y$ be any solution to (\ref{eqn:homdecrypt}). Since $y^Tf = 0$, we have that
\[ y^Tc = y^T(Mx + m{\bf 1}) = \Bigl(\sum\nolimits_{i \in S} y_iM_i\Bigr)x + m\Bigl(\sum\nolimits_{i \in S} y_i\Bigr) = m \]
which proves the claim.
\end{proof}

The next fact follows directly from the definitions of $\KK$ and $K$. 

\begin{fact}
The encryption scheme $\KK$ implements the encryption space $K$ with encryption error $\eta s$.
\end{fact}

\paragraph{Proto-homomorphic operations} We now define the notion of homomorphic and proto-homomorphic operations on ciphertexts, which plays an important role in homomorphic constructions.

\begin{definition}
Let $(Keys, Enc, Dec)$ be an encryption space with message set $\Sigma$ and ciphertext set $\Xi$. Let $\circ$ and $\circledcirc$ be binary operations on $\Sigma$ and $\Xi$, respectively. 
\begin{itemize}
\item We will say $\circledcirc$ is {\em homomorphic} for $\circ$ if for every $(PK, SK) \in Keys$ and $m, m' \in \F_q$,
\[ Enc_{PK}(m) \circledcirc Enc_{PK}(m') \subseteq Enc_{PK}(m \circ m'). \]
\item We will say $\circledcirc$ is {\em proto-homomorphic} for $\circ$ if for every $(PK, SK) \in Keys$ and $m, m' \in \F_q$,
\[ Enc_{PK}(m) \circledcirc Enc_{PK}(m') \subseteq Dec_{SK}(m \circ m'). \]
\end{itemize}
\end{definition}

Here, $\circledcirc$ is extended to an operation on sets in the natural way. The definitions extend naturally to unary operations. Now let $\oplus$ and $\odot$ denote pointwise addition and pointwise multiplication over $\F_q^n$ respectively, and let $\gamma\cdot\/$ denote multiplication of a vector in $\F_q^n$ by the fixed scalar $\gamma$.

\begin{claim}
\label{claim:proto}
With respect to the encryption space $K$, $\oplus$ is homomorphic for addition, $\gamma\cdot\/$ is homomorphic for multiplication by the scalar $\gamma$, and and $\odot$ is proto-homomorphic for multiplication.
\end{claim}
\begin{proof}
Let $c = Mx + m{\bf 1} + f$ and $c' = Mx' + m'{\bf 1} + f'$, where $f_i = f'_i = 0$ when $i \in S$. Then $c \oplus c' = M(x + x') + (m + m'){\bf 1} + (f + f')$, which is in $Enc_{PK}(m + m')$, proving homomorphism for additions. Scalar multiplications are similar. For multiplications, let $y$ be any solution to (\ref{eqn:homdecrypt}) and notice that
\begin{align*}
y^T (c \odot c') &= \sum\nolimits_{i = 1}^n y_i (Mx + m\bone + f)_i (Mx' + m'\bone + f')_i \\
&=  \sum\nolimits_{i \in S} y_i (Mx + m\bone)_i (Mx' + m'\bone)_i \\
&= \sum\nolimits_{i \in S} y_i (M_i \otimes M_i)^T(x \otimes x') + m \cdot y^TMx' + m' \cdot y^TMx + mm' \cdot y^T\bone \\
&= mm'  
\end{align*}
since by the constraints (\ref{eqn:homdecrypt}) we have $\sum_{i \in S} y_i (M_i \otimes M_i) = 0$, $y^TM = 0$, and $y^T\bone = 1$.
\end{proof}

Claim~\ref{claim:proto} already enables homomorphic evaluation under $\KK$ of circuits that have at most one layer of multiplication gates. To do more, we need a homomorphic way of turning ciphertexts of the form $Dec_{SK}(m)$ into ciphertexts of the form $Enc_{PK}(m)$. While we will not achieve this---at least not under the desired security assumption---in the following sections we will show how to convert $Dec_{SK}(m)$ into $Enc_{PK'}(m)$, where $PK'$ is a different public key. We describe this process of {\em reencryption} in the following section.

\section{Reencryption}
\label{sec:reenc}

We now define the functionality and security requirements of reencryption. We then prove a composition theorem which shows how to obtain homomorphic encryption from reencryption and a basis of proto-homomorphic operations.

Intuitively, a reencryption circuit takes a decryption under keys $(PK, SK)$ and outputs an encryption under keys $(PK', SK')$. To do this the circuit will access some auxiliary information about the secret key $SK$ which will be ``hidden'' under $PK'$. We model this auxiliary information by an {\em auxiliary key information} function $I(SK, PK')$. One complication that occurs in our instantiations of reencryption is that the function $I$ will be randomized, and we will have to account for the possibility that it produces incorrect information about the key pair.

\begin{definition}
Let $E = (Keys, Enc, Dec)$ and $E' = (Keys', Enc', Dec')$ be encryption spaces over the same message set. A (deterministic) circuit $\ReEnc_{I(\cdot)}(\cdot)$ is a {\em reencryption} from $E$ to $E'$ with auxiliary key information $I$ and key error $\kappa$ if for every admissible pair $(PK, SK) \in Keys, (PK', SK') \in Keys'$, every message $m$ and every $c \in Dec_{SK}(m)$,
\[ \pr_I[\ReEnc_{I(SK, PK')}(c) \in Enc_{PK'}(m)]  \geq 1 - \kappa \]
where the outer probability is taken only over the randomness of $I$.
\end{definition}

To define security, let $\EE$ and $\EE'$ be encryption schemes that implement $E$ and $E'$ respectively. We will say $\ReEnc$ is $(s \to s', \eps \to \eps')$-secure provided that for every pair of messages $m_1$ and $m_2$, if $(PK, \Enc_{PK}(m_1))$ and $(PK, \Enc_{PK}(m_2))$ are $(s, \eps)$ indistinguishable, then $(PK, PK', I(SK, PK'), \Enc_{PK}(m_1))$ and $(PK, PK', I(SK, PK'), \Enc_{PK}(m_2))$ are $(s', \eps')$ indistinguishable.

We now show how to combine proto-homomorphic operations and reencryption in order to obtain homomorphic encryption.
One small complication is that in our definition of reencryption we allow that the two schemes $\EE$ and $\EE'$ are different. This is an important feature that will help us achieve the definition initially. So when we apply $d$ levels of reencryption, we will work with a chain of public-key encryption schemes $\EE_0, \dots, \EE_d$.

Let $\EE_0, \dots, \EE_d$ be public-key encryption schemes so that $\EE_i$ implements encryption space $E_i$. Assume $\ReEnc_i$ is a reencryption from $E_i$ to $E_{i+1}$ with auxiliary information $I_i$.

Let $C$ be a circuit with binary gates, each of which has a homomorphic or proto-homomorphic implementation in all of the spaces $E_i$. Abusing terminology, we will call these gates homomorphic and proto-homomorphic gates, respectively.
The {\em proto-homomorphic depth} of $C$ is the largest number of proto-homomorphic gates on any directed path in any circuit in $\calC$. Without loss of generality (by adding some dummy gates), we will assume that the proto-homomorphic gates in $C$ are layered, i.e. every path in every circuit has exactly the same number of proto-homomorphic gates. Let $\calC^{\circ}_{cs, d}$ be the class of circuits of size $cs$ and proto-homomorphic depth $d$.

\vspace{\baselineskip}
\noindent{\bf Homomorphic template $\TT(\EE_0, \dots, \EE_d)$} for $\calC^{\circ}_{cs, d}$
\smallskip
\hrule
\smallskip
\noindent{\bf Key generation}: Generate key pairs $(PK_i, SK_i)$ uniformly at random for every $i$. Generate auxiliary key information $I_i(SK_i, PK_{i+1})$ uniformly at random for every $i$. The secret key is $(SK_0, SK_d)$. The public key is $(PK_0, \dots, PK_d, I_0, \dots, I_{d-1})$.

\smallskip
\noindent{\bf Encryption and decryption} are the same as in $\EE_0$ using the key pair $(PK_0, SK_0)$.

\smallskip
\noindent{\bf Homomorphic decryption} is the same as in $\EE_d$ using the secret key $SK_d$.

\smallskip
\noindent{\bf Homomorphic evaluation}: Given a layered circuit $C$, replace every homomorphic gate $+$ of $C$ by its homomorphic implementation $\oplus$. At every proto-homomorphic layer $i$, replace the proto-homomorphic gates $\cdot$ by their proto-homomorphic implementations $\odot$ followed by $\ReEnc_i$. Add reencryption gates $\ReEnc_0$ to the input level. Perform the evaluations of the ciphertext, using auxiliary information $I_i$ for $\ReEnc_i$. Output the resulting ciphertext.
\smallskip
\hrule
\vspace{\baselineskip}

The following two statements capture the functionality and security properties of this scheme; we omit the easy proofs.

\begin{proposition}
\label{prop:homfunc}
Suppose $\ReEnc_i$ has key error at most $\kappa$. Then $\TT(\EE_0, \dots, \EE_d)$ is a homomorphic encryption scheme with setup error at most $d \cdot\kappa$.
\end{proposition}

\begin{claim}
\label{claim:homsec}
Suppose $\EE_0$ is $(s_0, \eps_0)$-message indistinguishable and $ReEnc_i$ is $(s_i \to s_{i+1}, \eps_i \to \eps_{i+1})$ secure for every $i$. Then $\TT(\EE_0, \dots, \EE_d)$ is $(s_d, \eps_d)$-message indistinguishable.
\end{claim}

\subsection{Constructing reencryption}

We now give a construction of a reencryption from the family of encryptions $\KK_q(n)$. Let $\KK_q(n)$ and $\KK_q(n')$ be two instantiations of $\KK$ with a different hardness parameter, specifically with $n' > n$. To simplify notation we will identify the two encryption schemes with their corresponding encryption spaces.

Our construction of a reencryption from $\KK_q(n)$ to $\KK_q(n')$ is based on Gentry's ingenious idea of homomorphically evaluating the decryption circuit of $\KK_q(n)$. The decryption circuit in our scheme is extremely simple as it only uses homomorphic additions. However, one important complication in our scheme is the possibility of encryption errors. While for a single encryption the likelihood of an error occurring is small, when we apply the encryption to all the coordinates of the ``secret key'' the error becomes substantial. Our choice of parameters for $\KK_q(\cdot)$ is essential for controlling the error; it will allow us to tolerate a substantial amount of error provided we choose $n'$ to be large enough in terms of $n$.

We now describe the reencryption. Let $y$ be the designated solution to the system (\ref{eqn:homdecrypt}), which specifies the decryption space of $\KK_q(n)$. Recall that $y_i = 0$ whenever $i$ is outside the hidden subset $S$. The auxiliary key information $I(SK, PK')$ consists of the encryptions $z_1 = \Enc_{PK'}(y_1), \dots, z_n = \Enc_{PK'}(y_n)$, where all encryptions are performed independently. Each of these encryptions is a vector in $\F_q^{n'}$. The reencryption is given by
\[ \ReEnc_{z_1, \dots, z_n}(c) = c_1 z_1 + \dots + c_n z_n. \]

\begin{claim}
\label{claim:reencbasic}
$\ReEnc$ is a reencryption from $\KK_q(n)$ to $\KK_q(n^{1 + \alpha})$ with auxiliary information $I$ and key error $n^{-\alpha(1 - \alpha)/2}$.
\end{claim}
\begin{proof}
Recall that $z_i$ has the form $M'x_i + y_i\bone + e_i$, where $e_i$ is an error vector with error rate $\eta'$. We will say the output of $I(PK', SK)$ is {\em good} if for all $i \in [n]$, all the entries of $e_i$ that fall inside the hidden subset $S'$ are zero. By a union bound, the probability that $I(PK', SK)$ is not good is at most
\[ \eta' s' n = n^{-(1 + \alpha)(1 - \alpha/4)} \cdot n^{(1 + \alpha)(\alpha/4)} \cdot n = n^{-\alpha(1 - \alpha)/2}. \]
We now show that if $I(PK', SK)$ is good then $\ReEnc_I(c) \in Enc_{PK'}(m)$ for every $c \in Dec_{SK}(m)$. Recall that $Enc_{PK'}(m)$ contains those ciphertexts that take value $M'_{S'}x + m\bone$ inside $S'$ (for some $x$) and can take arbitrary value outside $S'$. Since $I$ is good, we know that the projection of $z_i$ onto $S'$ has the form $M'_{S'}x_i + y_i\bone$. Therefore the projection of $\ReEnc_I(c)$ to $S'$ has the form
\[ \sum\nolimits_{i = 1}^n c_i(M'_{S'}x_i + y_i\bone) = M'_{S'}x + (c^Ty)\bone = M'_{S'}x + m\bone \]
where $x = \sum c_ix_i$.
\end{proof}

The following security claim can be derived by a hybrid argument.

\begin{claim}
If $\KK_q(n')$ is $(s, \eps')$-message indistinguishable then $\ReEnc$ is $(s \to s - \poly(n), \eps \to \eps + n\eps')$-secure.
\end{claim}

Assume $\KK_q(n)$ is $(s, \eps(n))$-message indistinguishable for every $n$, where $\eps(n)$ is nonincreasing. Instantiating the template $\TT(\EE_0, \dots, \EE_d)$ with the encryption schemes $\EE_i = \KK_q(n^{(1 + \alpha)^i})$, we obtain a family of homomorphic encryption schemes $\BASIC(n)$ for circuits $C\colon \F_q^m \to \F_q$ with addition, scalar multiplication, and binary multiplication gates of size $cs$ and multiplication depth $d$ with key length and encryption length $O(n^{(1+\alpha)^d})$ and setup error $dn^{-\alpha(1 - \alpha)/2}$ that are $(s - d\cdot\poly(n), O(n^{(1+\alpha)^{d-1}}\eps(n)))$-message indistinguishable.

\section{Optimizing reencryption}
\label{sec:optimize}

We now describe two transformations to reencryption. The purpose of the first transformation is to eliminate the blowup in the security parameter in Claim~\ref{claim:reencbasic}. The second one is a generic technique for reducing the key error. 

\subsection{Improving the key length}
\label{sec:keylength}

Let us revisit the homomorphic scheme $\BASIC$ from the previous section. For convenience we will introduce a change of parameters. After performing $d$ layers of homomorphic multiplication, the length of the ciphertext went from $n_0$ to $n = n_0^{(1 + \alpha)^d}$. We will describe a reencryption from $\KK_q(n)$ to $\KK_q(n)$.

What we would like to do is use the transformation from Claim~\ref{claim:reencbasic}, but without increasing the length $n$. As we noted, this is difficult to do owing to the large amount of encryption error that accumulates into the auxiliary key information. Now let us attempt to {\em reduce} the reencryption length by moving from $\KK_q(n)$ to $\KK_q(n_0)$. This appears even less reasonable, as $\KK_q(n_0)$ has even greater encryption error than $\KK_q(n)$. But one advantage of working with $\KK_q(n_0)$ is that the scheme $\BASIC$ already allows us to do homomorphic evaluation over its ciphertexts. Our idea is to apply $\BASIC$ to a ``correction circuit'' $CORR$ whose purpose is to eliminate the encryption errors introduced when encrypting the secret key information about $\KK_q(n)$ using $\KK_q(n_0)$.

To carry out this idea, we have to be somewhat careful about the design of $CORR$. Here, the value of the parameter $\alpha$ will play an important role. If $CORR$ is too deep the security suffers, as it is dictated by  $n_0$, while the encryption length is $n \gg n_0$. For a careful choice of the parameters, we can ensure that $CORR$ has constant depth, which will enable us to produce length-preserving reencryptions of size $n$ with security parameter polynomial in $n$.

We will assume that $q$ is a power of two. Let $d$ be an even constant (we later set it to $8$). Let $(PK, SK)$ and $(PK', SK')$ be two admissible key pairs for $\KK_q(n)$. 

\medskip\noindent{\bf Reencryption.} We generate the auxiliary key information as follows. First, sample a sequence of independent key pairs $(PK_0, SK_0), \dots, (PK_{d-1}, SK_{d-1})$, where $(PK_i, SK_i)$ comes from $\Gen(n_0^{(1+\alpha)^i})$. Let $y \in \F_q^n$ specify the decryption space of $\KK_q(n)$. The auxiliary information is generated as follows. Let $\gamma$ be a generator for the field extension $\F_q$ over $\F_2$.
\begin{enumerate}
\item {\bf Encrypt:} For each coordinate $y_i$ of $y$, expand as $y_i = y_{i0} + \gamma y_{i1} + \dots\/ + \gamma^{\log q - 1}y_{i\log q - 1}$ with $y_{ij} \in \B$. For every $i, j$, create $2^d$ independent ciphertexts $c_{ij}^k = \Enc_{PK_0}(y_{ij})$, where $k$ ranges from $1$ to $2^d$.
\item {\bf Correct:} For every $i, j$, calculate $z_{ij} = \Eval(CORR, c_{ij}^1, \dots, c_{ij}^{2^d})$, where $\Eval$ is the evaluation algorithm for $\BASIC$ when the key generation algorithm is instantiated with the keys $(PK_0, SK_0), \dots, (PK_{d-1}, SK_{d-1}), (PK', SK')$, and $CORR\colon \B^{2^d} \to \B$ is the circuit described below.
\item {\bf Output:} Let $z_i = z_{i0} + \gamma z_{i1} + \dots\/ + \gamma^{\log q - 1}z_{i\log q-1}$. Output the vector $I(SK, PK') = (z_1, \dots, z_n)$. 
\end{enumerate}

As before, the reencryption procedure is $\ReEnc_{z_1, \dots, z_n}(c) = c_1z_1 + \dots + c_nz_n$.

We now describe the correction circuit. The purpose of this circuit is to eliminate the errors accumulated in the encryption, which suggests using majority. However we also need to have fine control over the depth of the circuit. Since the errors of various encryptions are independent, it is natural to use a recursive majority-type construction in order to correct the error from one layer to the next. For our analysis, it will be convenient to make $CORR$ be a full binary tree of depth $d$ where $d$ is even and all the gates are of the type $G(x, y) = 1 - xy$. When restricted over $\B$ inputs, this is a NAND tree.

\begin{proposition}
\label{prop:reencsame}
For $\alpha \leq 1/4$ and $d = 8$, $\ReEnc$ is a reencryption from $\KK_q(n)$ to $\KK_q(n)$ with auxiliary key information $I$ and key error $O(n^{-0.5})$.
\end{proposition}
\begin{proof}
With probability $dn^{-\alpha(1-\alpha)/2}$ over the choice of keys, we know that the circuit $\Eval$ makes no mistake on its input. Let us assume this is the case. 

We will show that with probability $1 - O(n^{-0.5})$, $z_{ij} \in Enc_{PK'}(y_{ij})$ for every pair $(i, j)$. By the homomorphic property of additions and scalar multiplications, it follows that $z_i \in Enc_{PK'}(y_i)$ for all $i$. The correctness of reencryption then follows by the same argument as in Claim~\ref{claim:reencbasic}.

We fix $i$ and $j$ and for notational convenience we write $y = y_{ij}$, $z = z_{ij}$, $c^k = c^k_{ij}$. Let $\hat{y}^k$ denote the unique value in $\F_q$ such that $\Dec_{SK_0}(c^k) = \hat{y}^k$. Since the encryption of the $y_{ij}$s was performed at error rate $\eta_0$, it follows that independently for each $y$, $\hat{y}^k = y$ with probability $1 - \eta_0$, and otherwise $\hat{y}^k$ could be an arbitrary element in $\F_q$.

Let us start with the special case $d = 2$. We will argue that the $\pr[z \not\in Enc_{PK'}(y)] \leq 6\eta_0^2$. This follows from the design of the circuit $CORR$. If $CORR$ is given four inputs, three of which have the same value $0$ or $1$, its output will also have the same value. Therefore the event $z \not\in Enc_{PK'}(y)$ can only happen if $\hat{y}^k \neq y$ for at least two values of $k$, which happens with probability at most $6\eta_0^2$. 

By induction on (even values of) $d$, it follows that in general the event $z \not\in Enc_{PK'}(y)$ can happen with probability at most $6^{2^{d/2} - 1}\eta_0^{2^{d/2}}$. We now take a union bound over all pairs $i$ and $j$ and conclude that the reencryption is correct with probability at least $n(\log q)(6\eta_0)^{2^{d/2}}$.

Now recall that $\log q \leq n$ and $n = n_0^{(1 + \alpha)^d}$, which gives an error of
\[ n_0^{2(1 + \alpha)^d} (6\eta_0)^{2^{d/2}} = \frac{6^{2^{d/2}}}{n_0^{(1 - \alpha/4){2^{d/2}} - 2(1 + \alpha)^d}} \leq \frac{6^{2^{d/2}}}{n_0^{(15/16) \cdot 2^{d/2} - 2\cdot (5/4)^d}} = O(n_0^{-3.07}) = O(n^{-0.5}) \]
for $d = 8$.
\end{proof}

The following claim follows by a standard hybrid argument and we omit the proof.

\begin{claim}
\label{claim:reencsecure}
Fix $\alpha \leq 1/4$ and $d = 8$ and assume $\KK_q(n)$ is $(s(n), \eps(n))$-message indistinguishable for every $n$, where $\eps(n)$ is nonincreasing. Then for every $\eps_0$, $\ReEnc$ is $(s(n) \to s(n^{0.1}) - \poly(n), \eps_0 \to \eps_0 + O(n^{1.8} \cdot \eps(n^{0.1}))$-secure.
\end{claim}

\subsection{Reducing the key error}
\label{sec:errorprob}

The final optimization we perform concerns the key error of reencryption. The key error of the reencryption $\ReEnc$ from the previous section cannot be reduced beyond $1/n$. In the homomorphic template in Section~\ref{sec:reenc}, the setup error increases linearly with the number of reencryptions, so we cannot apply this scheme to circuits of depth larger than $n$. We now introduce a generic technique for reducing this error.

Suppose we are given a reencryption $\ReEnc$ with key error $\kappa \leq 1/32$. If we apply $\ReEnc$ $k$ times in parallel to the same ciphertext but using independent instantiations of the auxiliary key information, by large deviation bounds we can expect that with probability $1 - 2^{-\Omega(k)}$, a significant majority---say a $15/16$ fraction---of the reencryptions will be correct. However, reapplying reencryption over and over again will quickly yield overwhelming error. This calls for a boosting tool of the following kind: Given $k$ ciphertexts out of which, say, $15/16$ represent the same value, output $k$ ciphertexts out of which a larger majority, say $31/32$, now represent that value. We implement this functionality in a circuit that we call $\Boost$. For later convenience we reencrypt the outputs of $\Boost$.

\begin{definition}
Let $E$ and $E'$ be two encryption spaces over the same message set and $(PK, SK)$, $(PK', SK')$ be a pair of admissible keys from the respective spaces. A {\em booster} of length $k$ from $E$ to $E'$ with auxiliary key information $I(SK, PK')$ and key error $\kappa$ is a circuit $\Boost$ with the following property. For every message $m \in \B$ and ciphertexts $c_1, \dots, c_k$ out of which at least $15k/16$ belong to $Dec_{SK}(m)$, $\Boost_{I(SK, PK')}(c_1, \dots, c_k)$ outputs ciphertexts $c'_1, \dots, c'_k$ out of which at least $31k/32$ belong to $Enc_{PK'}(m)$.
\end{definition}

We emphasize that we only require the definition holds for messages $m \in \B$, and not arbitrary messages in $\F_q$. The security definition for boosters is identical to the one for reencryptions.

Our construction of boosters is based on von Neumann's idea of robust evaluation of circuits with faulty gates~\cite{Neu56}. Let $G$ be a bipartite expander graph with $k$ vertices on each side. The circuit $\Boost$ will apply $G$ to its inputs and perform a homomorphic majority at each  output. Computing each of these homomorphic majorities may require some reencryptions. The auxiliary key information in each of these reencryptions will be independent, ensuring that with very high probability few errors will be introduced in the reencryption.

\paragraph{The construction} Assume $\EE$ is an encryption scheme equipped with $\oplus$, $\odot$ and reencryption $\ReEnc$ over ciphertexts of length $n$. Let $G$ be an $(n, b, \lambda = 1/32)$ spectral expander~\cite{HLW} for a sufficiently large constant $b$,  and let $APXMAJ_b\colon \F_q^b \to \F_q$ be a circuit of depth that depends only on $b$ (not on $q$) so that
\begin{equation}
\label{eqn:apxmaj}
APXMAJ_b(x_1, \dots, x_b) = \begin{cases}
0, &\text{if at least $7b/8$ of the inputs are $0$}, \\
1, &\text{if at least $7b/8$ of the inputs are $1$}.\end{cases}
\end{equation}
In Appendix~\ref{app:apxmaj} we show the existence of such a circuit of size $O(b^2)$ and depth $b' = O(\log b)$. 

\smallskip\noindent{\bf Auxiliary key information $I(SK, PK')$:} Repeat the following independently $b'$ times, once for every output $j$ of $\Boost$: First, generate a sequence of keys $(PK^j_1, SK^j_1), \dots, (PK^j_{b'-1}, SK^j_{b'-1})$ and set $SK = SK^j_0, PK' = PK^j_{b'}$. Output $I'(SK^j_i, PK^j_{i+1})$ for every $i$ and $j$, where $I'$ is the auxiliary key information for $\ReEnc$.

\smallskip\noindent{\bf The circuit $\Boost$:} Suppose that output $j$ of $G$ is connected to inputs $j_1, \dots, j_b$. For every output $j$, apply the homomorphic evaluation to the circuit $APXMAJ_b$ on inputs $c_{j_1}, \dots, c_{j_b}$ as described in Section~\ref{sec:reenc}, but using the auxiliary key information with superscript $j$, and with an extra round of reencryptions at the output.

\begin{proposition}
\label{prop:boostfunc}
Assume $\ReEnc$ is a reencryption whose key error $\kappa$ is a sufficiently small absolute constant (independent of $n$). Then $\Boost$ is a booster with key error $2^{-\Omega(k)}$.
\end{proposition}
\begin{proof}
By Proposition~\ref{prop:homfunc}, each of the homomorphic majority circuits has setup error at most $O(\kappa \log b)$. Since these setup errors are independent, by Chernoff bounds the chances that more than $k/64$ is at most $2^{-\Omega(k)}$. Let us assume this is not the case.

Now let $B$ be the set of inputs of $G$ whose value is different from $m \in \B$. By assumption, $\abs{B} \leq k/16$. Let $S$ be the set of outputs of $G$ that connect to more than $b/8$ inputs inside $B$. Then there are at least $\abs{S}b/8$ edges between $S$ and $B$. By the expander mixing lemma, $\abs{S}/8k \leq \abs{S}/16k + \lambda\sqrt{\abs{S}/16k}$, from where $\abs{S} \leq 16\lambda^2 k \leq k/64$ by our choice of $\lambda$. 

It follows that at most $k/64 + k/64 = k/32$ outputs of $\Boost$ will decrypt incorrectly with probability at most $1 - 2^{-\Omega(k)}$.
\end{proof}

We now state the security of this construction.

\begin{claim}
\label{claim:boostsecure}
If $\ReEnc$ is $(s \to s', \eps_0 \to \eps_0 + \eps)$-secure, then $\Boost$ is $(s \to s' - k \cdot \poly(n), \eps_0 \to \eps_0 + O(k\eps))$-secure.
\end{claim}

\section{The scheme $\HOM$}
\label{sec:hom}

To obtain our scheme $\HOM$, we will apply the homomorphic template of Section~\ref{sec:reenc} to $k$ parallel copies of the base scheme $\KK_q(n)$, using the booster from Section~\ref{sec:errorprob} to perform reencryptions. Let $n$ denote the security parameter.

Let $\KK^k_q(n)$ denote the following scheme over message set $\F_q$ and ciphertext set $\F_q^{kn}$. The key generation algorithm is the same as in $\KK_q(n)$. To encrypt a message $m$, we output $k$ independent encryptions of $m$ in $\KK_q(n)$. To decrypt a ciphertext $c_1\dots\/c_k$, we apply the decryption of $\KK_q(n)$ on each $c_i$ and output the most frequent answer.

Let $K = (Keys, Enc, Dec)$ denote the encryption space for $\KK_q(n)$ from Section~\ref{sec:proto}. We now define an encryption space $K^k = (Keys, Enc^k, Dec^k)$ for $\KK^k_q(n)$. We let $Enc_{PK}^k(m)$ consists of those ciphertexts $c_1\dots\/c_k$ for which $c_i \in Enc_{PK}(m)$ for at least $31k/32$ values of $i$. We let $Dec_{SK}^k(m)$ consists of those ciphertexts $c_1\dots\/c_k$ for which $c_i \in Dec_{SK}(m)$ for at least $15k/16$ values of $i$.

It is easy to see that if $K$ is an encryption space for $\KK_q(n)$ with encryption error $1/64$, then $K^k$ is an encryption 
space for $\KK_q^k(n)$ with encryption error $2^{-\Omega(k)}$. The error follows from a large deviation bound.

It is also easy to see that pointwise addition $\oplus$ and pointwise multiplication $\odot$ are proto-homomorphic over message set $\B$ with respect to $K^k$. Notice that although $\oplus$ was homomorphic for $K$, it is merely proto-homomorphic for $K^k$, owing to the possibility of erroneous encryptions in $Enc^k$.

Finally, notice that the booster $\Boost$ from Section~\ref{sec:errorprob} (instantiated with the length-preserving reencryption $\ReEnc$ from Section~\ref{sec:keylength}) is a {\em reencryption} for $K^k$. Now define

\[ \text{$\HOM = \TT(\KK^k_q(n), \dots, \KK^k_q(n))$ with reencryption $\Boost$} \]
where $\TT$ is the homomorphic template from Section~\ref{sec:reenc}. The following two claims prove Theorem~\ref{thm:main}.

\begin{claim}
The scheme $\HOM$ is a homomorphic encryption scheme for $\calC_{cs, d}$ with key length $O(dkn)$ and setup error $d \cdot 2^{-\Omega(k)}$.
\end{claim}

This claim follows directly froms Proposition~\ref{prop:homfunc} and Proposition~\ref{prop:boostfunc}.

\begin{claim}
Assume $\KK_q(n)$ (with $\alpha \leq 1/4$) is $(s(n), \eps(n))$-message indistinguishable, where $s(n)$ and $1/\eps(n)$ are nondecreasing. Then $\HOM$ is $(s(n^{0.1}) - dk\cdot\poly(n), O(dkn^{1.8}\eps(n^{0.1})))$-message indistinguishable.
\end{claim}

This claim follows by combining Claims~\ref{claim:homsec},~\ref{claim:reencsecure},~and~\ref{claim:boostsecure}.

\section{Conclusion}

In this work we propose a new public-key encryption system that is inspired by the conjectured hardness of decoding noisy codewords from certain affine codes with a planted trapdoor. We argue the security of this system and give a construction of a secure homomorphic encryption scheme based on it.

To evaluate a circuit of depth $d$, our scheme requires keys of size $O((d \log d) n)$, where $n$ is the security parameter. It would be good if this dependence of $d$ in the key length was eliminated. One important tool in our analysis is the length-preserving reencryption circuit from Section~\ref{sec:reenc}. There we proved that reencryption is secure provided it is used on independent key pairs. It is tempting to instantiate this construction over the same key pair, in the spirit of ``circular security'' prevalent in other works on homomorphic encryption. This would indeed eliminate the dependence on $d$ (and also obviate the need for reducing the key error).

While we do not know if the suggested circular security assumption is valid or not, we are uncomfortable conjecturing it for the following reason. In the auxiliary key information, every one of the $n$ elements $y_i$ of the ``secret key vector'' $y$ is encoded by a ciphertext $c_i$ of length $n$, so that all the ciphertexts decode without error. In view of the simplicity of our decryptions, we feel that if such a property holds at all, it should be achievable by direct construction (possibly using other reasonable security assumptions) rather than the somewhat complex mechanism of Section~\ref{sec:reenc}. We were not able to come up with such a direct construction without suffering a security flaw.

Our initial motivation for this research was to better understand the complexity required for homomorphic encryption. Owing to the simplicity of its encryption, the scheme of Applebaum et al. was a natural starting point for this study. Many of the techniques developed here can be applied to that scheme. However, we were unable to design a secure length-preserving reencryption for that scheme. In short, the reason is that the system of equations analogous to (\ref{eqn:homdecrypt}) for that scheme does not enjoy a sufficient amount of redundancy, which severely limits the choice of $\alpha$.

We recently learned of an independent attempt by Armknecht et al.~\cite{AAPS11} to construct a code-based homomorphic encryption scheme. Their scheme achieves only some rudimentary homomorphic properties and is not public-key. However it appears some of their ideas (for example, the use of pointwise operations on ciphertexts) are related to ours and it would be interesting to see if they can be applied towards future improvements.

\bibliographystyle{alpha}
\bibliography{fhecoding}

\appendix

\section{The ranks of submatrices of the public key}
\label{app:fullrank}

We prove the following proposition, which points to the limitation of an attack on the public key of $M$ described in the introduction.

\begin{proposition}
Let $T \subseteq [n]$, $\abs{T} = t$ be an arbitrary subset of rows of the $r \times n$ public key matrix $P$ such that $\abs{T \cap S} \leq s/3 + \max\{t - r, 0\}$. Then the submatrix $P_T$ of $P$ spanned by the rows indexed by $T$ has full rank with probability at least $1 - O(r^2/q)$, where the randomness is taken over the choice of $a_1, \dots, a_n$ in the key generation algorithm. 
\end{proposition}
\begin{proof}
We prove the theorem for the matrix $M$ instead of $P$. Since $P$ and $M$ have the same column space and the rank of $P_T$ is a property of the column space of $P$ projected to the coordinates in $T$, the statement will follow.

Without loss of generality we may assume that $M_T$ is a square matrix: If $t < r$ we can augment the $M_T$ by rows from outside $S$, and if $t > r$, we can eliminate rows from $M_T$ that come from $S$ (and some extra ones if necessary). Both operations preserve rank deficiency.

Now suppose $M_T$ is a square matrix so that at most $s/3$ of its rows come from $S$. Let us assume, again without loss of generality, that $T = \{1, \dots, r\}$ and $S = \{1, \dots, s_0\}$, $s_0 \leq s/3$. We now argue that with probability $1 - O(r^2/q)$, the determinant $\det(M_T)$ is nonzero.

Notice that $\det(M_T)$ is a formal polynomial in the variables $a_1, \dots, a_r$ of degree at most $1 + 2 + \dots + r = r(r+1)/2$. In our setup, the diagonal term $a_1a_2^2\dots\/a_r^r$ appears uniquely in the sum-product expansion of the determinant, and so this formal polynomial is nonzero. By the Schwarz-Zippel lemma, if $a_1, \dots, a_r$ were chosen independently at random from $\F_q$, $\det(M_T)$ would be zero with probability at most $1 - r(r+1)/2q$. Our $a_i$ are not independent since they are required to be distinct, but the statistical distance between $r$ uniformly independent elements of $\F_q$ and $r$ uniform but distinct elements of $\F_q$ is only $O(r^2/q)$. It follows that $\det(M_T) \neq 0$ with probability $1 - O(r^2/q)$.
\end{proof}

\section{Approximate $0, 1$-majorities over arbitrary fields}
\label{app:apxmaj}

In this section we prove the following claim.

\begin{proposition}
Let $q$ be the power of a prime. There exists a circuit $APXMAJ_m\colon \F_q^m \to \F_q$ of size $O(m^2)$ and depth $O(\log m)$ with the property (\ref{eqn:apxmaj}).
\end{proposition}

The challenge is to make the depth of the circuit independent of $q$. We show an easy construction based on a trick of Valiant~\cite{Val84}. 

\begin{proof}
Let $CORR_d$ be the correction circuit from Section~\ref{sec:keylength} where $d = 2\log m + 4$. We will show that there exists a way to connect the $m$ inputs to the $2^d$ inputs of $CORR_d$ in a way that the resulting circuit computes $APXMAJ_m$.

Fix a specific input $x$ so that at least $7/8$ of its elements equal $b$. If each of the inputs to $CORR_d$ is randomly wired to one of the elements in $x$, then the inputs to $CORR_d$ will take value $b$ independently with probability at least $7/8$ each. Recall that for $b \in \B$, if each of the inputs to this circuit takes value $b$ with probability $7/8$, then its output takes value $b$ with probability $1 - (3/4)^{2^{d/2}} > 1 - 2^{-m}$ by our choice of $d$. Taking a union bound over all such inputs $x$, we conclude that there must exist a wiring with the desired property.
\end{proof}

\end{document}